\documentclass[10pt,journal,compsoc]{IEEEtran}
\usepackage{amsmath}
\usepackage{amssymb}
\usepackage{algorithmic}
\usepackage{cite}	
\usepackage{graphicx}
\usepackage{caption}
\usepackage{subcaption}
\usepackage{amsthm}
\usepackage{mathtools}
\usepackage{algorithm}
\usepackage{algorithmic}
\newtheorem{thm}{Lemma}
\newtheorem{prop}{Proposition}
\newcommand\blfootnote[1]{%
	\begingroup
	\renewcommand\thefootnote{}\footnote{#1}%
	\addtocounter{footnote}{-1}%
	\endgroup
}

\begin{document}
%\title{Coded Caching for Heterogeneous Users' Behaviours }
\title{Content Caching for Shared Medium Networks Under Heterogeneous Users' Behaviours }

\author{\IEEEauthorblockN{Abdollah Ghaffari Sheshjavani\IEEEauthorrefmark{1}, 
Ahmad Khonsari\IEEEauthorrefmark{1}\IEEEauthorrefmark{2},	
Seyed Pooya Shariatpanahi\IEEEauthorrefmark{1}, 
and Masoumeh Moradian\IEEEauthorrefmark{2}\\ 
%\IEEEauthorblockA{\IEEEauthorrefmark{1}School of Electrical and Computer Engineering, University of Tehran, Iran}
%\IEEEauthorblockA{\IEEEauthorrefmark{2}School of Computer Science, Institute for Research in Fundamental Sciences (IPM), Iran}
%Emails: \{abdollah.ghaffari, a\_khonsari, p.shariatpanahi\}@ut.ac.ir, mmoradian@ipm.ir
}
\thanks{\IEEEauthorblockA{\IEEEauthorrefmark{1}School of Electrical and Computer Engineering, College of Engineering, University of Tehran, Iran}}
\thanks{\IEEEauthorblockA{\IEEEauthorrefmark{2}School of Computer Science, Institute for Research in Fundamental Sciences (IPM), Iran}}
\thanks{Emails: \{abdollah.ghaffari, a\_khonsari, p.shariatpanahi\}@ut.ac.ir, mmoradian@ipm.ir}
}

\maketitle

\begin{abstract}
\fontdimen2\font=0.54ex
Content caching is a widely studied technique aimed to reduce the network load imposed by data transmission during peak time while ensuring users' quality of experience.  
It has been shown that when there is a common link between caches and the server, delivering contents via the coded caching scheme can significantly improve performance over conventional caching. However, finding the optimal content placement is a challenge in the case of heterogeneous users' behaviours. In this paper we consider heterogeneous number of demands and non-uniform content popularity distribution in the case of homogeneous and heterogeneous user-preferences. We propose a hybrid coded-uncoded caching scheme to trade-off between popularity and diversity. We derive explicit closed-form expressions of the server load for the proposed hybrid scheme and formulate the corresponding optimization problem.  
Results show that the proposed hybrid caching scheme can reduce the server load significantly and outperforms the baseline pure coded and pure uncoded and previous works in the literature for both homogeneous and heterogeneous user preferences.
\end{abstract}

% no keywords

\blfootnote{Some parts of this paper is the extended version of the problem that is presented at the WCNC 2020 conference\cite{9120820}.}

%\IEEEpeerreviewmaketitle

\section{Introduction}

\subsection{Background}
\fontdimen2\font=0.54ex
The global increase in the penetration of high throughput wireless devices such as tablets and smartphones has significantly facilitated the growing demand for mobile content through wireless media in recent years. Deploying small base stations (SBSs) to increase spatial reuse of the frequency spectrum by shrinking the network cells size is a promising solution to alleviate this growth and has stimulated many research initiatives \cite{chandrasekhar2008femtocell}. Nonetheless, the high cost of  wired links and the bottleneck of wireless links still poses as the main obstacle to providing high-speed backhaul to connect SBSs to the core network in this approach. To address this problem in content delivery scenarios, caching popular contents at these SBSs has been proposed to relieve the need for high-speed backhaul links \cite{shanmugam2013femtocaching, paschos2016wireless, chen2017probabilistic}.

In general, conventional caching methods attempt to cache the most popular contents located at close proximity of end-users such that the requests for popular contents are directly served from the local caches. This results in the so-called local caching gain, which is proportional to the local memory size. 
In \cite{maddah2014fundamental}, the authors introduced a novel coded caching scheme
 that significantly improves performance over conventional caching by leveraging the multicasting nature of the shared (such as wireless) medium even for caches with distinct demands. Their scheme, in addition to the local caching gain, results in global caching gain through using coded-multicast opportunities. The global caching gain 
 is proportional to the aggregate memory of all the caches, where every user benefits from its cache contents to decode the desired content and to remove the interference in the coded message due to other caches requests. This idea has been further generalized to hierarchical coded caching \cite{karamchandani2016hierarchical}, multi-server coded caching \cite{shariatpanahi2016multi}, decentralized coded caching \cite{maddah2015decentralized, 7742405, 7999228}, online coded caching\cite{7055939}, device to device (D2D) coded caching  \cite{7342961}, hybrid server-D2D coded caching \cite{8904142}, coded caching with
 asynchronous user requests \cite{8374865},
and coded caching with multiple file requests \cite{wei2017coded,7133172, 8594642}.

To increase multicasting opportunities in coded caching, diverse parts of the library should be cached among different users, i.e. the \emph{diversity principle}. However, in a set-up with non-uniform content popularity distribution, it is desirable to cache more popular contents with higher frequency, i.e. the \emph{popularity principle}, which makes the cache contents of different users almost the same. As these two principles, namely diversity and popularity, are in tension, cache placement design in such scenarios is very challenging. %In addition, previous studies indicate that the global popularity cannot be directly used to infer the local popularity of contents \cite{7875170}.

Heterogeneity can affect on the tension between diversity and popularity by affecting on multicasting opportunities. We can divide heterogeneity in caching into two main architectural and users' behavioral categories. Architectural heterogeneity includes different content and cache sizes and unequal users' downloading rate \cite{8017548, 8943165, 8704184,8977539}. Users' behavioral heterogeneity includes different users' preferences\cite{8849597 ,8761540} and different numbers of users' requests at each time slot \cite{multilevel2017}. 

%\subsection{Main Contributions}
 In this paper, we investigate the content caching in a shared medium network and propose a hybrid coded-uncoded caching under heterogeneous users' behaviors in order to minimize the shared medium traffic volume. we consider a caching scheme which partitions the contents into three groups; coded-cached, uncoded-cached, and non-cached ones. 
 In particular, we first focus on the problem of coded caching for non-uniform user-independent (homogeneous-i.e., it is identical for all users) content popularity distribution where each user may request multiple contents in each query. Then, we generalize this problem to non-uniform user-dependent (heterogeneous) content popularity distribution. In both cases, we derive explicit closed-form expressions of the shared medium traffic for the proposed hybrid coded caching scheme. In fact, this scheme proposes the optimal trade-off between popularity (uncoded caching gain) and diversity (coded caching gain). 
 
 In practice, the proposed heterogeneous caching scenario corresponds to a cellular network that includes a Macro Base Station (MBS) and multiple SBSs where each SBS is equipped with a limited size cache and serves multiple users. In this regard, on one hand, the number of requests each SBS sends to MBS in each query (or time slot), depends on the number of users it serves. On the other hand, the number of requests for each content depends on the popularity distribution over the SBS coverage area. 
 
Finally, the numerical and simulation results show that the proposed hybrid caching outperforms the baseline pure coded, pure uncoded and previous works as well as the two-partitioning scheme reported in 
\cite{hachem2015effect, li2017traffic, ji2017order, zhang2018coded, 8863425} for both SBS-independent and SBS-dependent content popularity distributions.

\subsection{Organization}
The rest of the paper is organized as follows. In Section~\ref{sec2} an overview of the coded caching is provided and related works are reviewed. In Section~\ref{sec3}, the system model is introduced. The proposed caching schemes for multiple requests with SBS-independent and SBS-dependent non-uniform demands is described in Section~\ref{hetero-semi} and ~\ref{hetero-hetero}, respectively. 
This is followed by numerical analysis and simulation results in Section~\ref{sec6}. Finally, Section~\ref{sec7} concludes the paper.
\section{ Background of Coded Caching and Related Works}
\label{sec2}
In this section, we first summarize the coded caching scheme reported in \cite{maddah2014fundamental} and then review the subsequent related works.
\subsection{Background on Coded Caching}
The authors in \cite{maddah2014fundamental} consider a system with one server connected through a shared, error-free link to $K$ users. The server access to the database of $N$ contents each of size $F$ bits. Each user is equipped with a cache memory of size $MF$ bits. Their system operates in two phases: a $placement$ phase and a $delivery$ phase. 
In the placement phase,  %of the original coded caching scheme \cite{maddah2014fundamental},
each content is split into ${K}\choose{T}$ non-overlapping equal-sized sub-files, where $T=K \!\times\! M / N$ and the size of each sub-file
is equal to $F/{{K}\choose{T}}$. The sub-files are distributed at caches such that each cache stores $M/N$ of each content. Moreover, each sub-file has $T$ copies in $T$ different caches. In the delivery phase, each cache receives a request for a single content. The server then XORs the required sub-files by different caches according to a specific coding strategy and multicasts coded messages to the corresponding groups of $T+1$ caches. The achievable rate of the coding strategy for serving all contents at the shared link is proven to be \cite{maddah2014fundamental}:
\begin{equation}
R = K \left(1-\frac{M}{N}\right) \min \left\{\frac{1}{1+K\times M/N},\frac{N}{K}\right\}.
\label{eq2}
%\vspace{-0.3em}
\end{equation}
%\color{red} 
Where $K (1-\frac{M}{N})$ is the local caching gain and $\frac{1}{1+K\times M/N}$ is the global caching gain.
%\color{black}
\subsection{Related Works}
Although the original coded caching scheme introduced in \cite{maddah2014fundamental} performs well under homogeneous systems, the scheme is inefficient in non-uniform and heterogeneous content popularity and also heterogeneous architectural scenarios (\cite{niesen2017coded, hachem2015effect, li2017traffic, ji2017order, zhang2018coded, multilevel2017,8943165, 8704184, 8977539, 8849597, 8761540,8863425}).
In the non-uniform content popularity (user-independent), different contents have different popularity but the popularity of any particular content is the same for all users (\cite{niesen2017coded, hachem2015effect, li2017traffic, ji2017order, zhang2018coded, multilevel2017}).  
On the other side, in the user-dependent content popularity, in addition to the fact that the popularity of contents is not the same, also the popularity of each particular content is not the same for different users \cite{8849597, 8761540}. User-dependent content popularity is also called user preference in the literature.

To handle coded caching for non-uniform content popularity scenarios, one major approach in the literature is grouping contents based-on their popularity\cite{niesen2017coded, hachem2015effect, li2017traffic, ji2017order, zhang2018coded, multilevel2017}.  For the first time, authors in \cite{niesen2017coded} proposed a grouping method to address non-uniform content popularity. In their method, in the placement phase, the library is partitioned into almost equiprobable groups and each user's cache is evenly shared among these groups. Finally, each group is treated as a single coded caching problem originally proposed in \cite{maddah2014fundamental}. The efforts in \cite{hachem2015effect, li2017traffic, ji2017order, zhang2018coded, 8863425} show that the asymptotically optimum placement strategy of grouping method of \cite{niesen2017coded} is to partition the library into two groups: the popular contents are cached according to the scheme in \cite{maddah2014fundamental} while the non-popular contents are not cached at all. 
However, authors in \cite{ multilevel2017} show by some examples that when each cache receives multiple requests, in the grouping method of \cite{niesen2017coded}, partitioning the library altogether into three groups improves caching performance over two-partitioning placement strategies. In their method, the first part is cached fully at all the cache memories, and the second part is cached according to the original coded caching paradigm and the last part is not cached at all. Nevertheless, the authors in \cite{multilevel2017} assume that the library is divided into multiple levels, based on varying degrees of popularity. Besides, this work does not consider the closed-form expressions for the optimum partitioning under arbitrary popularity distribution. 

Some works, such as \cite{8226776} show that for uniform popularity distribution, the caching strategy proposed in \cite{maddah2014fundamental} can be improved by removing sending some redundancy in the delivery phase. For example, suppose we have one content $A$ of size $F$ bits, three users, and a caching size of $1/3$$\times F$ bits for each user. Based-on the placement strategy of \cite{maddah2014fundamental}, $A$ is divided into 3 pieces $A_1, A_2, A_3$ which each piece cached in the corresponding user. In the delivery phase, the scheme in \cite{maddah2014fundamental} suggests to broadcast $A_1 \oplus A_2$, $A_1 \oplus A_3$ and $A_2 \oplus A_3$ that provide a delivery rate of 1. But $A_2 \oplus A_3$ can be recovered from  $(A_1 \oplus A_2) \oplus ( A_1 \oplus A_3)$, and therefore we do not need to broadcast it and the delivery rate reduces to $2/3$. 
Authors in \cite{8769911} generalized this idea to propose a new coded caching strategy under uncoded placement to handle non-uniform demands. This strategy uses equal sub-packetization for all contents while allowing to allocate more cache to more popular content. However, they propose and analyze the delivery strategy only for the case of the existing two contents in the system.

Other approaches, such as \cite{ji2015efficient,8761540} used a structured clique cover algorithm for all demands of the users to handle coded caching with non-uniform content popularity and other aspects of heterogeneity. In \cite{ji2015efficient}, the authors introduced a coded caching scheme where the users have more than one request, and the content popularity is non-uniform. They used a random popularity-based algorithm for cache placement and adapting the idea of dividing each content into equal-sized sub-contents. Then, subsequently they used a greedy constrained local graph coloring technique to find multicast opportunities in the delivery phase. Authors in \cite{8761540} consider heterogeneous user preferences. In this work, each user caches its most probable content at the placement phase, then in the delivery phase, based on the request and cached matrix, tries to gain from multicast opportunities. Therefore, these works do not consider any optimization for the placement phase.

  The architectural heterogeneity in coded caching is considered in some works, such as\cite{8017548, 8943165, 8704184, 8977539}. Authors in \cite{8017548} studied the coded caching with unequal link rates and proposed the use of nested coded modulation (NCM) coding in the delivery phase. However, the main drawback of this work is that the cache size for each user needs to be correctly allocated to adapt NCM transmission, in a way that users with lower link rate need a larger cache size. In\cite{8977539}, the authors analyzed the coded caching problem in a generalized scenario of the D2D coded caching \cite{7342961}, where the cache size of users is unequal. In addition, coded caching under non-uniform file-length, non-uniform users cache size, and non-uniform content popularity is considered in \cite{8943165}. This work shows that finding optimal caching with the three aforementioned heterogeneity has exponential complexity. Therefore, they developed a tractable optimization problem corresponding to a caching scheme with the above three  heterogeneities and showed numerically that it performs well compared to the original exponentially scaling problem. Authors in \cite{8704184} also considered the coded caching problem under non-uniform users' cache size and download rate. Some works, such as \cite{8291755, 8503143} considered heterogeneous quality-of-service requirements in which each content may have various resolution copy based on the different users' device resolution requirements. However, these works do not consider the heterogeneity of user behaviors, such as user-dependent content popularity, while, previous studies indicate that the global popularity cannot be directly used to infer the local popularity of contents \cite{7875170}.

 Finally, the authors in \cite{8849597} considered some aspects of two categories of heterogeneity, such as heterogeneous content sizes, heterogeneous cache size, and user-dependent content popularity. Although this work considers full heterogeneous content popularity, it analyzes the problem only for a very small scale such as a two users / two files scenario.

\section{System Model}
\label{sec3}
We consider a cellular network that consists of one MBS, which is connected through a shared error-free link to $K$ SBSs, as depicted in \figurename{~\ref{fig1}}. The content library has $N=\{W_1, W_2, ... , W_N\}$ distinct contents that are all accessible by the MBS.  Without loss of generality, we assume that all contents have the same size equal to $F$ bits. Each SBS has a cache memory of size $M \times F$ bits for some integer number of $M \in [0, N]$, and SBS $c$ is responsible for serving $Z_c$ users where $c \in \{1, 2, \dots, K\}$. Each user can connect and receive data from the MBS and only one SBS.

\begin{table}
	\centering
	\small
	\caption{Summery of the main notation.}
	\label{table:notation}
\begin{tabular}{|l|l|}%{l l}
	\hline 
	Symbol & Explanation \\ 
	\hline 
	$n$ & generic content  \\ 
	$c$ & generic SBS (cache)  \\ 
	$i$& generic step of sending coded contents \\
	$N$ & number of contents  \\ 
	$F$ & size of each content(bit)  \\
	$K$ & number of SBS \\ 
	$M$ & cache capacity of each SBS (content)  \\
	$Z_c$ & number of users (demands) in range of SBS $c$ \\ 
	$Z_{max}$ & $\max (\{Z_c\}_{c=1}^K)$\\
	$D_c$ & demand vector of SBS $c$ \\
	$q_{c,j}$ & probability of requesting the $j$th distinct\\&
	coded content at the next request in SBS $c$ \\
	$p_{n,c}$ & popularity of content $n$ in SBS $c$ \\
    $q_c^{coded}$ &  the queue of the coded requests of SBS $c$\\
	$q^{uncoded}$ &  the queue corresponding to uncoded requests\\
	$P^{(c)}_i$&  probability of at least $i$ distinct requests\\ & 
	in $q_c^{Coded}$ at step $i$\\
	$Q_i$ & number of non-empty coded queues  at step $i$\\
	$l_{c}$ &  number of distinct request in the $q_c^{coded}$\\
	$Q_i^{(c)} $& number of non-empty coded queues  in first $c$ \\ &
	caches at step $i$\\
	$l^{(z)}_{c}$ &  number of distinct coded request in first $z$\\ &
	 requests that received by SBS $c$\\
	$g$ & generic group of SBSs\\
	$Y_{n,c}$& indicate content $n$ is (or not) cached in SBS $c$ \\ 
	$X_{n,g}$& indicate content $n$ is (or not) cached in group $g$\\
	$S_{c,g}$& indicate SBS $c$ is (or not) participate in group $g$ \\
	$r_1$& traffic load (MBS) for the coded contents\\ 
	$r_1^{(i)}$& $r_1$ at step $i$\\ 
	$r_2$& traffic load (MBS) for un-cached contents\\ 
	$r$& total traffic load of the MBS $(r_1+r_2)$\\ 
	\hline 
\end{tabular} 
\end{table}
Similar to previous works, our system operates in two phases: the content placement phase and the content delivery phase. The placement phase is carried out during off-peak times. In this phase, the caching strategy determines some functions of all contents $S_c=f_c(W_1, ..., W_N)$, $c \in \{ 1, 2, \dots, K\}$ that must be cached at each SBS based on the system parameters such as cache memory constraint and content popularity, then the caches are filled with corresponding contents from the library.  

In the delivery phase, only the MBS has access to the whole library. Moreover, each SBS receives one request from each user connected to it within its range, resulting in a total of $Z_c$ requests. Therefore, each SBS may receive multiple requests for some contents of the library. Denote $D_c=[d_{1,c}, ... , d_{Z_c, c}]$ as the demand vector of SBS $c$, where $d_{i,c}$ is the content requested by user $i$. Moreover, the number of distinct contents in $D_c$ can be between $1$ and $Z_c$ (due to the possibility of duplicate requests for contents by different users). Upon collecting the requests of the users in SBSs, the MBS receives the list of distinct requested contents from each SBS and then, sends information of size $R$ bits over the shared link to satisfy these requests. 
The content popularity distribution is arbitrary and can be SBS-dependent. We specify $p_{n,c}$ as the probability of requesting the content $W_n$  by  users under coverage of SBS $c$ where $n \in \{1,2, \dots, N\}$, and  $c \in \{1, 2, \dots, K\}$. 

In the delivery phase, the network load may be high, and the MBS suffers from low available bandwidth. Thus, our objective is to design a content caching scheme that minimizes the traffic load on the shared link, which is the bottleneck in the delivery phase. However, as we show in section ~\ref{hetero-hetero}, finding the optimal placement strategy of coded caching that minimizes the traffic load is intractable in the case of heterogeneous content popularity distribution across different SBSs (i.e., SBS-dependent content popularity). 

 In this paper, we first consider designing of an optimal cache placement for the case of SBS-independent content popularities in Section ~\ref{hetero-semi}, as the assumption of homogeneous popularities has been adopted by many previous works. %(please give some ref.s here)
 Then, we generalize the same problem to the case of SBS-dependent one, in Section ~\ref{hetero-hetero}. 

\begin{figure}[!t]
	\centering
	\includegraphics[width=3.2in]{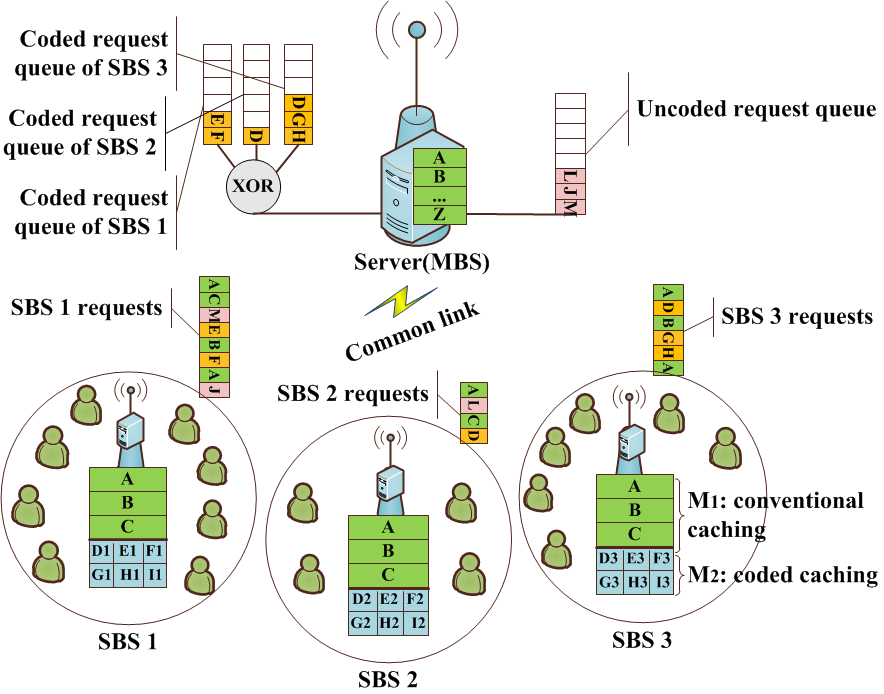}
	\caption{An example of our caching system with one MBS containing $N=26$ contents of size $F$ bits connected via an error-free shared link to $K=3$ SBSs, each with a cache size of $M \!\times\! F=5F$ bits and serving $Z_c$ end-users.}
	\label{fig1}
	\vspace{-0.5em}
\end{figure}
\section{SBS-independent non-uniform content popularity with non-uniform multiple demands  }
\label{hetero-semi}
 In this section, we assume that the content popularity distribution is the same for all SBSs. Therefore $p_{n,c} = p_n$, where $c \in \{1, 2, \ldots ,K\}$.
 
 As mentioned before, previous works such as \cite{hachem2015effect, li2017traffic, ji2017order, zhang2018coded, 8863425} show that dividing the contents into two groups is order-optimal for the coded caching scheme under non-uniform demands. Besides, \cite{ multilevel2017}, shows by some examples, that when each cache receives multiple requests, partitioning the library altogether into three groups improves caching performance over two-partitioning strategies. However, this work does not propose the closed-form expressions for optimizing partitioning. In this regard, in this section, we extend our previous grouping method of caching under non-uniform demands in \cite{9120820} to non-uniform multiple requests and formulate the optimization problem for minimizing the load of the shared link by presenting the hybrid coded-uncoded caching scheme. 
 
 \subsection{The proposed SBS-independent caching scheme}
%As mentioned before, similar to previous works, our proposed method has two phases: the placement phase and the delivery phase. 
In the placement phase of the proposed caching scheme, we categorize the contents into (at most) three groups based on their request probabilities. In particular, we first choose the $N_1$ most popular contents among all $N$ contents and then cache $M_1$ most popular contents among these $N_1$ selected ones entirely at all caches. Then, the remaining $N_1 \!-\! M_1$ contents are cached using the coded caching scheme proposed in \cite{maddah2014fundamental}. Accordingly, each cache memory is divided into two parts: $M_1 \!\times \!F$ bits of each cache are allocated to the $M_1$ most popular contents, while the remaining $(M \!-\! M_1) \!\times F$
bits of memory are allocated to the coded caching scheme with a library size of $N_1 \!- M_1$ contents. In summary, the three groups of contents resulting from our scheme are:
\begin{enumerate}
	\item  $M_1$ most popular contents that are cached completely.
	\item  $(N_1 \!-\! M_1)$ popular contents that are cached according to \cite{maddah2014fundamental}.
	\item $(N \!-\! N_1)$ least popular contents that are not cached at all.
\end{enumerate}

In the delivery phase, each SBS receives $Z_c$ number of requests. The cache of each SBS locally serves the content requests belonging to the first group, whereas the MBS server is responsible for the requests belonging to the second and third groups. 
In order to perform the coded caching scheme, the MBS has to maintain the requests for coded contents of each SBS separately. In this regard, As shown in \figurename{~\ref{fig1}}, the MBS owns $K$ distinct queues, where $q_c^\text{coded}$, $c  \in \{1, 2, \ldots, K\}$ denotes the queue that stores the requests of SBS $c$ for coded contents. Moreover, requests are stored in $q_c^\text{coded}$ by an arbitrary ordering. In addition, the MBS has one single queue which stores the requests of all SBSs for uncoded contents,
denoted by $q_{\text{uncoded}}$. 

We now explain the steps involved in the transmission of coded messages. 
Initially, MBS collects all head-of-line (HoL) requests of the queues $q_c^\text{coded}$ ($c  \in \{1, 2, \ldots, K\}$) and then, in order to response to these requests, transmits the corresponding coded messages, following the scheme in\cite{maddah2014fundamental}. 
 The MBS then updates the HoL requests in the queues and repeats the same procedure, i.e., in step $i$, the $i$th rows of all queues are considered by the coded scheme. Note that the number of requests in $q_c^\text{coded}$ could be less than $Z_c$ since some of the requests of SBS c belong to either the first or the third group and thus are not stored in $q_c^\text{coded}$.
 Even requests belonging to the second group of contents might be repetitive and thus, are not stored in $q_c^\text{coded}$ separately. Therefore, the number of queues involved in the coding process at step $i$ could be less than $K$ since some of the queues may not have any requests at step $i$. Moreover, the number of steps is at most $\max_{c\in\{1,\dots, K\}}\{Z_c\}$, which happens when all requests of the SBS with the maximum number of users are distinct and associated with the coded contents. Finally, after sending all requested contents belonging to the second group with the coded scheme, the MBS sends contents related to all requests in $q_{\text{uncoded}}$, which guarantees that all users will be able to retrieve their requested contents.

\figurename{~\ref{fig1}} depicts a scenario where the number of SBSs is $K = 3$. The SBSs $1$, $2$, and $3$ are responsible for $8$, $4$, and $6$ users, respectively, and thus, they receive $Z_1 = 8$, $Z_2=4$, and $Z_3 = 6$ requests, respectively.
The total number of contents is $N = 26$ and are ordered based on their popularity (i.e., `A' is the most popular content). The cache size is $M = 5$ contents, and we assume that $M_1 = 3$, and $N_1 = 9$. As a result, the contents `A', `B', and `C' are cached entirely. Moreover, $6$ contents (from `D' to `I') are cached based on the coded caching scheme. The remaining $17$ less popular contents (from `J' to `Z') are not cached. In the delivery phase, the uncoded data is highlighted in pink, while the contents in yellow are to be transmitted in the coded manner. The contents in green are locally hit by the local cache of SBSs, without any further transmission from the MBS. In this figure, the HoLs of all the queues during the first, second, and third steps in coded transmissions consist of $k_1 = 3$, $k_2 = 2$, and $k_3 = 1$ content requests, respectively.

\subsection{Performance analysis}
\label{sec5}
In this sub-section, we determine the expected MBS traffic load as a function of $M_1$ and $N_1$. We then characterize the optimum partitioning strategy as an optimization problem to find the minimum load.
%\color{red}
In the following, we assume that contents are sorted according to their popularities, i.e.  $p_i \geq p_j$ if $i \leq j$, (contents with a lower index are more popular). 
Also, the traffic load contributed by the MBS is either related to the requests belonging to the second group of contents (i.e. coded contents with index from $M_1 \!+\! 1$ to $N_1$) or associated with requests belonging to the third group of  contents (i.e. uncoded contents with index from $N_1 \!+\! 1$ to $N$). In this regard, in the following, we first derive the traffic related to each of these groups, separately, and then formulate the optimization problem.  

In regard with the process explained in Section~\ref{sec3},  we denote $Q_i$ to be the random variable denoting the number of non-empty queues at step $i$, where $i \!=\! 1, 2,\ldots, \max (Z_1, ... ,Z_K)$. In the following, the calculation of the MBS traffic load is presented.
\begin{thm}
\label{lemma1}
The traffic load of the coded contents at step $i$ given that $Q_i = k$, denoted by $r_1^{(i)}$, is derived as:
\begin{align}
 r_1^{(i)} =\min\Bigg(&\frac{{{K}\choose {T+1}}-{{K-k} \choose {T+1}}}{{{K} \choose{T}}}, N_1 - M\Bigg).
  \end{align}
  where $T=\frac{K \times (M-M_1)}{(N_1-M_1)}$.
\end{thm}
\begin{proof}
See Appendix A for the proof.
\end{proof}
 In the coded scheme proposed by \cite{maddah2014fundamental}, in the delivery phase, the coded messages are sent to all subsets of size $T+1$ of users which has requested a content. In Lemma \ref{lemma1}, the term ${{K-k}\choose{T+1}}$ excludes those subsets that none of their members has requested a coded content. To calculate these subsets, in the following, we propose a lemma for deriving the probabilities that SBS $c$ has more than $i$ distinct coded requests. 
\begin{thm}
	\label{PIC}
	Let $P_{c,i}$ denote the probability that the number of distinct coded requests of SBS $c$, denoted by $l_{c}$, is equal  or greater than $i$, i.e., $P_{c,i} =Pr\{l_{c}\ge i\}$. $P_{c,i}$ is derived as follows:
	 \begin{align}
	&P_{c,i}  =  \sum_{j=i}^{Z_c}{Pr\{l_{c}=j\}} ,
	\label{lem2-4}
	\vspace{-0.8em}
	\end{align}
	assume the $ \text{Pr}\{ l_{c}^{(z)} = j \} $ to be the probability of having $j$ distinct coded requests in the first $z$ requests in SBS $c$, where $z \!=\! 1, 2,\ldots, Z_c$, then:
	\begin{align}
	&\text{Pr}\{ l_{c}=j \}= \text{Pr}\{ l_{c}^{(Z_c)} = j \}
	\end{align}
	$\text{Pr}\{ l_{c}^{(Z_c)} = j \}$ can be calculated with below recursive formula:
	\begin{align}
	&1) {Pr}\{ l_{c}^{(0)} = 0 \}=1 ,  {Pr}\{ l_{c}^{(z)} = j | j>z \}=0,%\\if $k$ equals to $0$:
	\nonumber \\
	& 2) {Pr}\{ l_{c}^{(z)} = 0 \} ={Pr}\{ l_{c}^{(z-1)} = 0 \} \times (1-q_{c,1}),
	\nonumber \\    %q^{c}_{1}
	&3) {Pr}\{ l_{c}^{(z)}=j\} = {Pr}\{ l_{c}^{(z-1)}=j \} \times (1-q_{c,j+1})
	\nonumber \\
	& +{Pr}\{ l_{	c}^{(z-1)}=j-1 \} \times q_{c,j},
	\label{lem2-7} 
	\end{align}
	where $q_{c.j}$ is the probability that $j$-th distinct coded content is requested at SBS $c$, and is derived as follows (in this section we assume $p_{n,c}=p_n$, $\forall c \in \{1,\ldots,K\}$):
	%\\if $j>N_1-M_1$ then $q_{c,j} =0$, otherwise it approximate as following 
	\begin{align}
	q_{c,j} = q_{j} \left\{ \begin{array}{l}
	=0,  \quad if \quad j>N_1-M_1,\\ 
	\scalebox{.92}{$\simeq (1- \dfrac {j-1}{N_1-M_1} )\times \sum_{n=M_1+1}^{N_1}{p_n},  \quad otherwise. $}
	\end{array}\right.
	%q_{j,c} = q_{j} \simeq (1- \dfrac {j-1}{N_1-M_1} )\times \sum_{n=M_1+1}^{N_1}{p_n}%}
	\label{lem2-8}
	\end{align}
\end{thm}
\begin{proof}
	See Appendix B for the proof.
\end{proof}
Using the above lemma, we derive the distribution of non-empty queues at each step of content delivery phase, where the HoL coded requests are responded.
\begin{thm}
\label{PrQi}
Assume $Pr\{Q_i=k\}$ denotes the probability that exactly $k$ SBSs request for coded contents at step $i$. It is derived as follows:
\begin{align}
%&\text{Pr}\{ Q_i = k \} = \text{Pr}\{ Q_i^{(K)} = k \} 
&\text{Pr}\{ Q_i = k \}= \text{Pr}\{ Q_i^{(K)} = k \}
\label{lem2-2}
\end{align}
where $Q_i^{(c)}$ is the random variable denoting the number of non-empty queues among the first $c$ queues, i.e., {$q_1^{coded}$, \dots, $q_c^{coded}$}, at step $i$ of coded caching.   
%If $N_1=M$, then $r_{1}$ equals to zero. 
${Pr}\{ Q_i^{(K)}= k \}$ can be calculated with the following recursive equations:
%${Pr}\{ Q_i^{(K)}= k \}$ is the solution of the following recursive equations: 
\begin{align}
&1) {Pr}\{ Q_i^{(0)}= 0 \}=1, {Pr}\{ Q_i^{(c)}= k | k>c \}=0,  %\\if $k$ equals to $0$:
\nonumber \\
& 2) {Pr}\{ Q_i^{(c)} = 0 \} ={Pr}\{ Q_i^{(c-1)} = 0 \} \times (1-P_{c,i}),
\nonumber \\
&3) {Pr}\{ Q_i^{(c)} = k \} = {Pr}\{ Q_i^{(c-1)} = k \} \times (1-P_{c,i})
\nonumber \\
& +{Pr}\{ Q_i^{(c-1)} = k-1 \} \times P_{c,i},
% &\text{Pr}\{ Q_i = k \} = {K \choose k} (P_{Zi})^{k} (1-P_{Zi})^{K-k}, 
\label{PrQiK} 
\end{align}
where $P_{c,i}$ is derived from Lemma \ref{PIC}.
\end{thm}
\begin{proof}
	See Appendix C for the proof.
\end{proof}	
In the following proposition, we derive the traffic load of MBS.
\begin{prop}
\label{proposition2}
Define $Z_{max}= \max_{c} Z_c$, %and at step $i$ the ${Pr}\{ Q_i^{(c)} = k \}  $ to be the probability of having $k$ caches with coded requests in first $c$ caches where $c \!=\! 1, 2,\ldots, K$. 
then the expected traffic load of coded content requests, denoted by $r_1$ is approximated by: \vspace{-0.4em}
\begin{align}
  \label{lem2-r1}
  r_{1} \simeq 
  \begin{cases}
  \scalebox{.95}{$\sum_{i=1}^{Z_{max}} \frac{{{K}\choose {T+1}}-\sum_{k=0}^{K}Pr\{Q_i=k\}{{K-k}\choose {T+1}}}{{{K} \choose{T}}}, \quad if\ N_1>M,$} \\
  0,\quad otherwise%if \ N_1=M.
  \end{cases}   
\end{align}
%where
\\
	where $Pr\{Q_{i}=k\}$ is derived from lemma \ref{PrQi}.
	Moreover, the expected traffic load of the uncoded content requests, denoted by $r_2$, is:\vspace{-0.4em}
\begin{align}   
	r_{2}=\sum_{n={N_1}+1}^{N} 1-(1-p_{n})^{\sum_{c=1}^{K}Z_c}.
	\label{proposition2-r2}
	\vspace{-0.4em}
\end{align}
Finally, the total expected traffic rate is $r = r_1+r_2$.
\end{prop}
\begin{proof}
From Lemma~\ref{lemma1}, the expected traffic load of the coded requests can be written as:\vspace{-0.4em}
\begin{align}
	r_1 = E\Bigg[&\sum_{i=1}^{Z_{max}}\min\Bigg(\frac{{{K}\choose {T+1}} \!-\! {{K-Q_{i}} \choose {T+1}}}{{{K} \choose{T}}}, N_1 \!-\! M\Bigg)\Bigg].
   \label{eq9}
\end{align}
Note that the expectation is taken over the random variables $Q_i$ for $i = 1, 2, \ldots, Z_{max}$. In \eqref{eq9}, the minimization function involves two terms; the first term is maximized at $Q_i = K$ as follows:
\begin{align}
 	\max\left(\frac{{{K}\choose {T+1}}-{{K-Q_{i}} \choose {T+1}}}{{{K} \choose{T}}}\right) = \frac{{{K}\choose {T+1}}}{{{K} \choose{T}}} = \frac{K-T}{T+1}.
 	\label{eq10}
\end{align}
By replacing $T$ in the above equation, we have:
  	\begin{align}
 		\frac{K-T}{T+1} = \frac{K \times (N_1-M)}{N_1-M_1 + K \times (M-M_1)}.
 	\label{eq11}
 	\end{align}

Letting \eqref{eq11} be less than $N_1-M$, i.e., the second argument in the min function in \eqref{eq9}, leads to: %If \eqref{eq11} is less than $N_1-M$, then the first term in the minimization function of \eqref{eq9} is always selected. To this end, we should thus have:
\begin{align}
		1 < \frac{N_1-M_1}{K}+(M - M_1), \nonumber
\end{align}
which always holds in the case of $N_1>M$ and $M-M_1 \geq 1$. Thus, the first argument of min in \eqref{eq9} is selected.
Also, in the case that $N_1 = M$, and therefore $M_1=M$ (pure uncoded caching), the second argument of min in \eqref{eq9} is selected, and thus we have $r_1 = 0$. Therefore, when  $N_1>M$ (and so $M_1<M$ ), \eqref{eq9} reduces to:\vspace{-0.4em}
\begin{align}
		r_{1} &= E\left[\sum_{i=1}^{Z_{max}} \frac{{K \choose T+1}-{K-Q_{i} \choose T+1}}{{K \choose T}}\right]  \nonumber \\
		&= \sum_{i=1}^{Z_{max}} \frac{{K \choose T+1} - E\left[{K - Q_i \choose T+1}\right]}{{K \choose T}} \nonumber \\
		&= \sum_{i=1}^{Z_{max}} \frac{{ K \choose T+1} - \sum_{k=0}^K Pr\{Q_i=k\}{K-k \choose T+1}}{{K \choose T}}.
		\label{r1}
		\vspace{-0.4em}
\end{align}

Next we prove \eqref{proposition2-r2}. As mentioned before, any request (in all $K$ SBSs) from the third group of contents that is not cached (contents indexed from $N_1 + 1$ to $N$) should be satisfied directly by the MBS. However, if the MBS receives multiple requests for specific content in a time slot, it uses broadcasting to send the content only once. Hence, the expected traffic load of such uncached contents is equal to the expected number of distinct requests for them. The probability that content $n$ is not requested is $(1-p_n)^{\sum_{c=1}^{K}Z_c}$. Therefore, the probability that content $n$ is requested at least one time is $1 - (1 - p_n)^{\sum_{c=1}^{K}Z_c}$. Thus, the expected total number of distinct requests of uncoded contents is equal to the sum of aforementioned expected probability for all uncoded contents, as indicated in \eqref{proposition2-r2}. This completes the proof.
\end{proof}
We now formulate the optimum partitioning problem in order to minimize the traffic load from MBS to the SBSs, i.e., $r$. %In other words, we need to find the optimum $M_{1}$ and ${N_1}$ values that minimize $r$. 
The minimization problem is formulated as follows:%,  $T= \frac{K \times (M-M_1)}{(N_1-M_1)}\in \mathbb{N}$:  
\begin{align}
&\min \limits_{\substack{ M \leq N_1 \leq N  \\  0 \leq M_1 \leq M}} \!\{
r_1+r_2 \}
\nonumber \\
&s.t. \nonumber \\
&T=\frac{K \times (M-M_1)}{(N_1-M_1)}\in \mathbb{N}.
\label{optimization1}
\end{align}
  
%As can be seen in the \eqref{optimization1}, 
If $N_1>M$ and $M_1=0$, then \eqref{optimization1} is reduced to the optimization of the two-partitioning pure coded scheme. Also, it can be proved that when the popularity of contents is the same for all SBSs, then caching the $M$ most popular contents in SBSs is the optimal placement strategy of the pure uncoded scheme. In other words, if $N_1=M$, then \eqref{optimization1} is reduced to the optimal pure uncoded scheme. 

 To find the optimal placement strategy of the hybrid caching scheme, we adopt an enhanced exhaustive search algorithm to solve \eqref{optimization1} in polynomial time. To reduce the complexity of the search algorithm, we consider the following simplifications:
\begin{itemize}
	\item  
	We only consider those values of $N_1$ and $M_1$ which result in  $T = \frac{K \times(M-M1)}{(N1-M1)}$ to be an integer value.
	\item  By considering $\frac{{K \choose T+1} }{{K \choose T}}= \frac{K-T}{T+1}$, we calculate  $\frac{K-T}{T+1}$ only once for each possible $N_1$ and $M_1$ values.
	\item  As mentioned before, by  defining an array and storing the result of earlier computations (dynamic programming), we can compute \eqref{PrQiK} and \eqref{lem2-7} in polynomial time. We can also use this technique to skip duplicate computations in other parts of the problem; for example, storing the result of \eqref{proposition2-r2} for one value of $N_1$ could be used to compute it for the next $N_1$ values with less computing.
\end{itemize}
\section{SBS-dependent non-uniform content popularity with non-uniform multiple demands }
\label{hetero-hetero}
  \begin{figure}[!t]
	\centering
	\includegraphics[width=3.2in]{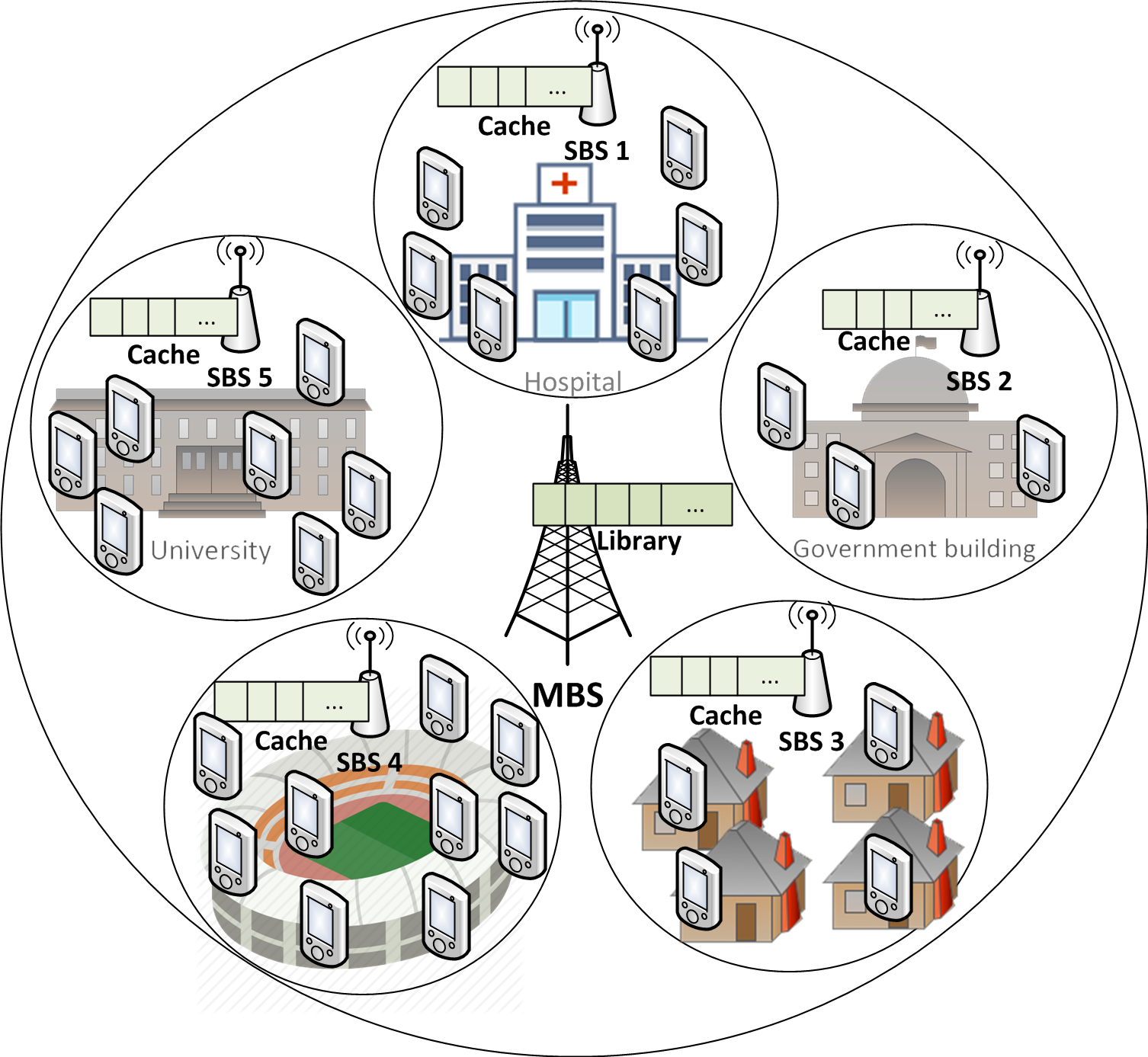}
	\caption{An example of SBS-dependent (heterogeneous) content popularity.}
	\label{heterogeniousModel}
	\vspace{-0.5em}
\end{figure}
In the previous section, we assumed that the content popularity distribution is the same for all SBSs. However, in the real world, users with similar preferences usually gather in the same locations, and there may be several popularity groups inside the coverage area of an MBS. Some contents may be popular for almost all groups of users, however, some other contents may be popular only for specific groups of them. 
As an example,  \figurename{~\ref{heterogeniousModel}},  illustrates an MBS with five SBSs under its coverage, where the SBSs serve groups with different preferences. If the coventional uncoded caching is used in such a scenario, i.e., the SBSs cache the files which are globally popular in their caches, then MBS has to send the local popular contents of SBSs repeatedly, leading to high traffic loads. However, if the SBSs store their local popular files, then MBS only transmits the least popular contents, taking advantage of broadcasting. Also in the case of applying coded caching, choosing files which are only popular in one SBS and not other SBSs leads to the waste of cache memory in other SBSs, and consequently increases the traffic load. Therefore, it is crucial to consider the heterogeneous SBS-dependent popularities for designing the caching strategy. On the other hand, SBSs may share similar preferences, e.g., in  \figurename{~\ref{heterogeniousModel}}, the users of SBS 4 are interested in sports, politic events are more popular in SBSs 2 and 5, and users of SBS 3 follow both politics and sports. Based on these similarities, we can form two groups of SBS, a group composed of SBSs 3 and 4 and another one composed of SBSs 2,3, and 5. %In this regard, we form different clusters of SBSs in this section to improve the performance of caching.

In this section, we extend the hybrid caching scheme of the previous section to the case of SBS-dependent popularities. In particular, unlike the previous section, where the parameters $N_1$ and $M_1$ are the same for all SBSs, in this section, they are optimized considering possibility of different groups.
\subsection{The proposed SBS-dependent caching scheme}
In this part, we define the caching strategy in the case of heterogeneous popularities, i.e., when $p_{n,c}$ is not the same for every SBS $c$, as follows. The cache of each SBS $c$ is divided into two part; uncoded and coded. The capacity of uncoded part in SBS $c$ is denoted by $M_{1c}$, i.e., $M_{1c}$ files are cached entirely in the uncoded part of SBS $c$. We let $Y_{n,c}=1$ if file $n$ is cached uncoded at SBS $c$ and  $Y_{n,c}=0$ otherwise. Consequently, we have $\sum_{n} Y_{n,c}=M_{1c}$, and $M-M_{1c}$ of cache capacity is left for coded part. On the other hand, we define groups of SBSs, where the SBSs within a group share similar preferences and thus, participate in the same coded caching scheme. It is worth noting that a single SBS may participate in multiple groups since while a part of its preferences are common in a group, other parts may be popular in other groups. In this regard, the grouping scheme of SBSs are defined as follows. We define $G = \{G_1,G_2,...,G_{|G|}\}$ to be a cover of $S$ (the set of SBSs),  where $G_{i} \subset S$ with $|G_i|\geq 2 $ and ${\cup}^{|G|}_{i=1} G_{i} = S$. According to this definition, we have $|G|$ groups of SBSs in cover $G$ which they have at least two members (for the sake of applying coded caching scheme), may be overlapping, and cover all SBSs. Moreover, we let $S_{c,g}=1$ if SBS $c$ participate in group $g \in G$, and $S_{c,g}=0$ otherwise. Then, the number of SBSs in group $g$, denoted by $K_g$, is derived as $\sum_{c}  S_{c,g}=K_g$.

A separate coded caching scheme is applied in each group. Since SBS $c$ may participate in multiple  groups, it dedicates  a part of its cache to each of its participating group.  In particular, the capacity $M_g$ is dedicated to group $g$, if $S_{c,g}=1$. Moreover, we let $X_{n,g}=1$ if file $n$ participates in the coded scheme of group $g $ and $X_{n,g}=0$, otherwise. Consequently, the number of files participating in coded scheme of group g, denoted by $N_g$, is derived as $\sum_n X_{n,g}=N_g, \forall g$. 

It is worth noting that in the proposed caching strategy for the heterogeneous case, unlike the homogeneous case, each SBS caches $M_{1c}$ uncoded contents that are different from other SBSs. Also, while the dedicated cache capacity for coded contents remains the same inside a group, but in general in each SBS, a different capacity is dedicated to coded contents. As such, the dedications in each SBS should satisfy $M_{1c}+\sum_{g} S_{c,g} \times M_g = M$. Moreover, since $|G|$ concurrent coded schemes are applied, MBS should maintain the different set of coded queues for each of these groups. In this regard, the queues $q^{coded}_{1,g},\dots, q^{coded}_{K_g,g}$ are the queues dedicated to group $g$. Finally, a specific content $W_n$ may be cached uncoded in SBS $c$, while it also be included in coded scheme of some of the participating groups. 
\subsection{Performance analysis}
In the following, we derive the traffic load of the MBS under the proposed caching strategy for SBS-dependent non-uniform content popularity distribution.
\begin{thm}
\label{PIgC}
Let $P_{c,g,i}$ denote the probability that the number of distinct coded requests of SBS $c$ in the cluster $g$, denoted by $l_{c,g}$, is equal or greater than $i$, i.e., $P_{c,g,i}=Pr\{l_{c,g}\geq i\}$. $P_{c,i,g}$ is derived as follows:
\begin{align}
&P_{c,g,i}=  \sum_{j=i}^{Z_c}{Pr\{ l_{c, g}=j\}}.
\label{lem5-4}
\end{align}
Also, let $ \text{Pr}\{ l_{c, g}^{(z)} = j \} $ be the probability of having $j$ distinct coded requests in first $z$ requests in SBS $c$, where $z \!=\! 1, 2,\ldots, Z_c$, then:
\begin{align}
&\text{Pr}\{ l_{c, g}=j \}= \text{Pr}\{ l_{c, g}^{(Z_c)} = j \},
\end{align}
where $\text{Pr}\{ l_{c, g}^{(Z_c)} = j \}$ can be calculated with below recursive formula:
\begin{align}
&1) {Pr}\{ l_{c, g}^{(0)} = 0 \}=1 ,  {Pr}\{ l_{c, g}^{(z)} = j | j>z \}=0,
\nonumber \\
& 2) {Pr}\{ l_{c, g}^{(z)} = 0 \} ={Pr}\{ l_{c, g}^{(z-1)} = 0 \} \times (1-q_{c,g, 1}),
\nonumber \\   
&3) {Pr}\{ l_{c, g}^{(z)}=j\} = {Pr}\{ l_{c, g}^{(z-1)}=j \} \times (1-q_{c,g, j+1})
\nonumber \\
& +{Pr}\{ l_{c, g}^{(z-1)}=j-1 \} \times q_{c,g, j},
\label{lem5-7} 
\end{align}
where $q_{c,g, j} $ is approximated to be:
%\\if $j>N_g$ then $q_{c,g, j} =0$, otherwise it approximate as following: 
\begin{align}
q_{c,g, j} =\left\{ \begin{array}{l}
=0,  \quad if \quad j>N_g,\\ 
\simeq (1- \dfrac {j-1}{N_g} )\times \sum_{n=1}^{N}{ X_{n, g} .p_{n,c}}\quad otherwise. 
\end{array}\right.
\label{lem5-8}
\end{align}
\end{thm}
\begin{proof}
See Appendix D for the proof.
\end{proof}	
\begin{thm}
	\label{lem_Qig}
Let $Pr\{Q_{i, g}=k\}$ be the probability that exactly $k$ SBSs in cluster $g$ request for coded contents at step $i$ . Then, we have:
\begin{align}
&\text{Pr}\{ Q_{i, g}= k \}= \text{Pr}\{ Q_{i, g}^{(K_g)} = k \},
\label{lem5-2}
\end{align}
where $Q_{i, g}^{(c_g)}$ is the random variable denoting the number of non-empty queues among the first $c_g$ queues in cluster $g$, i.e., {$q_{1,g}^{coded}$,..,$q_{c_g,g}^{coded}$}, at step $i$ of coded caching.   
${Pr}\{ Q_{i, g}^{(K_g)}= k \}$ can be calculated with the following recursive equations:
\begin{align}
&1) {Pr}\{ Q_{i, g}^{(0)}= 0 \}=1, {Pr}\{ Q_{i, g}^{(c_g)}= k | k>c_g \}=0,  %\\if $k$ equals to $0$:
\nonumber \\
& 2) {Pr}\{ Q_{i, g}^{(c_g)} = 0 \} ={Pr}\{ Q_{i, g}^{(c_g-1)} = 0 \} \times (1-P_{c_g,g,i}),
\nonumber \\
&3) {Pr}\{ Q_{i, g}^{(c_g)} = k \} = {Pr}\{ Q_{i, g}^{(c_g-1)} = k \} \times (1-P_{c_g,g,i})
\nonumber \\
& +{Pr}\{ Q_{i, g}^{(c_g-1)} = k-1 \} \times P_{c_g,g,i},
% &\text{Pr}\{ Q_i = k \} = {K \choose k} (P_{Zi})^{k} (1-P_{Zi})^{K-k}, 
\label{heteroPrQiK} 
\end{align}
where $P_{c_g,g,i}$ is derived from lemma \ref{PIgC}.
\end{thm}
\begin{proof}
The proof of this lemma is similar to the proof of lemma \ref{PrQi}, except that in this lemma, there are $G$ coded delivery clusters. Therefore, the equations are calculated for each cluster separately. However, for each cluster, only the SBSs that participate in it are considered.
\end{proof}	
\begin{prop}
\label{prop2}
	If $T_g=\frac{K_g \times M_g}{N_g}$, and $Z_{max}= \max_{c} Z_c$, then the expected traffic load of coded content requests, denoted by $r_1$, is  approximated by:
	\begin{align}
	&r_{1} =  \sum_{g \in G} \times \nonumber \\ 
	&
	\begin{cases}
	 \scalebox{.95}{$\sum_{i=1}^{Z_{max}}\frac{{{K_g}\choose {T_g+1}}-\sum_{k=0}^{K_g}Pr\{Q_{i, g}=k\}{{K_g-k}\choose {T_g+1}}}{{{K_g} \choose{T_g}}},\quad if\ N_g>M_g$}, \\
0,\quad otherwise
		\end{cases}   
	\label{prop2-r1}
	\end{align}
	where $Pr\{Q_{i, g}=k\}$ is derived from lemma \ref{lem_Qig}.
	Moreover, the expected traffic load of uncoded content requests, denoted by $r_2$, is:\vspace{-0.4em}
	\begin{align}   
	&r_{2}= \sum_{n=1}^{N} \Bigg(1- \nonumber \\ 
	&
	 \displaystyle\prod_{c=1}^{K}\Big(1-(p_{n,c}\times (1-Y_{n,c})\times \displaystyle\prod_{g \in G} (1-X_{n,g}. S_{c,g}) )\Big)^{Z_c}\Bigg).
	\label{prop2-r2}
	\vspace{-0.4em}
	\end{align}
Finally, the total expected traffic rate is $r = r_1+r_2$.
\end{prop}
\begin{proof}
Equation \eqref{prop2-r1} of this proposition is similar to equation \eqref{lem2-r1} of proposition \ref{proposition2}, except that in this proposition, there are $G$ coded delivery clusters. Therefore, the coded rate of the MBS is calculated for each cluster separately. Finally, the total coded rate of the MBS is the summation of calculated coded rates of all clusters. 
In order to prove  \eqref{prop2-r2}, we should consider that if content $W_n$ is cached neither uncoded at SBS $c$, i.e., $Y_{n,c}=0$, nor coded, i.e., $X_{n,g} \times S_{c,g}=0, \forall g \in G$, then content $W_n$ will be responded by MBS if it is requested in SBS $c$. 
 Consequently, the probability that content $W_n$ is not requested by SBS $c$ from MBS is equal to:  $\Big(1-(p_{n,c}\times (1-Y_{n,c})\times \displaystyle\prod_{g \in G} (1-X_{n,g}. S_{c,g}) )\Big)^{Z_c}$.
  If at least one SBS requests $W_n$ then, the file will be broadcast by MBS once and thus contributes to the traffic $r_1$.This completes the proof.	
\end{proof}
Let $\widetilde{G}$ denote the set of all possible covers of SBSs ($|\widetilde{G}|= 2^K-K-1$).Then the optimum partitioning problem with the objective of minimizing the traffic load from the MBS to SBSs is written as follows:
\begin{align}
&%r^{*}(N, M, K, \{Z_c\}_{c=1}^K, \{p_{n,c}\}_{n,c=1}^{N,K}) = %\nonumber \\
\min \limits_{\substack{
		G \subset \widetilde{G}
	  }} \{\min \limits_{\substack{
0 \leq M_{1c} \leq M \\
%\forall g \in G: 
1 \leq M_g \leq M \\
M_g < N_g \leq N  \\
X_{n, g} , Y_{n,c} \in \{0,1\}
%0 \leq K_g \leq K   
}}\{r_1+r_2\}\} 
\nonumber \\
&s.t. \nonumber \\
%&G=|\widetilde{G}|\\
&  T_g= \frac{K_g \times M_g}{N_g} \in \mathbb{N}, 
%\nonumber \\&S_{c,g}, X_{n, g} , Y_{n,c} \in \{0,1\}, 
\nonumber \\ &\sum_{n=1}^{N}{X_{n, g} }=N_g, \forall g \in G. %-M_1, 
%\nonumber \\ &\sum_{k=1}^{K}{S_{c,g} }=K_g, \forall g \in G. 
\nonumber \\& \sum_{n=1}^{N}{Y_{n,c} }=M_{1c} , \forall c \in \{1,\ldots,K\}. 
\nonumber \\& M_{1c}+\sum_{g\in G}{M_g \times S_{c,g} }=M , \forall c \in \{1,\ldots,K\}. 
\label{optimization2}
\end{align}
%Despite calculating $r$ based on the proposition \ref{prop2} has polynomial complexity for a specific configuration, 
As can be seen in \eqref{optimization2}, finding the optimal placement strategy for the hybrid scheme is intractable. Even if the optimal covering ($G$), the memory allocations ($M_g$) and the number of contents involved in each group ($N_g$) are specified then in order to find the best content placement, we need to calculate the MBS rate $r$ for $\displaystyle\prod_{c=1}^{K}{N \choose M_{1c}}$ $\times$ $\displaystyle\prod_{g=1}^{G} {N \choose N_g}$ possible configuration. 
A special case of this problem is when there is no significant similarity in the content popularity among the SBSs. In this case, increasing the number of coded delivery groups only reduces the global cache gain, and therefore the optimal cover $(G^\star)$ will have one member which is the set of all SBSs.
 %i.e., $\sum_{c=1}^{K}{S_{c,g} }=K$.  
Here, the optimal configuration is similar to the previous section, in which the SBSs' caches are divided into two parts. However, the difference is that the uncoded contents of one SBS can be different from other SBSs. Also, although the coded contents for all SBSs are the same, they are not necessarily the same as those would be selected based on global popularity. Besides, to find the best content placement for this special case, we need to calculate the MBS rate $r$ for ${N \choose M_1}^K \times {N \choose N_1}$ possible configuration. Even  in the cases of the two-partitioning pure coded,  \cite{hachem2015effect, li2017traffic, ji2017order, zhang2018coded, 8863425} and pure uncoded schemes we need to check ${N \choose N_1}$ and ${N \choose M}^K$ possible configurations respectively to find the optimal content placement.   
\section{Numerical Results }%and Discussions}
\label{sec6}
In this section, the performance of the proposed hybrid scheme is evaluated and compared with the baseline pure coded and conventional uncoded schemes and previous works, as well as the two-partitioning scheme reported in 
\cite{hachem2015effect, li2017traffic, ji2017order, zhang2018coded, 8863425}, through numerical evaluations and simulation. We validate analytical results with the simulations conducted using MATLAB for a period of $2000$ time slots. 
%Our experiments are run on a computer equipped with 8 GB of RAM and an AMD FX -4100 CPU operating at 3.6 GHz with 4 physical cores. However, our runs on MATLAB were non-parallel and consumed 25\% of the CPU.
In the following, we first evaluate the proposed hybrid scheme for SBS-independent non-uniform popularity distribution under a heterogeneous number of demands, and then we consider the SBS-dependent popularity distribution. 
\subsection{%heterogeneous number of demands,
	SBS-independent non-uniform popularity distribution}
   In this subsection, for analytical results, the optimum placement strategy of the hybrid scheme for integer values of $T = K \times (M-M_1)/(N_1-M_1)$ is obtained from \eqref{optimization1}. Also, we suppose that the content popularity distribution follows the Zipf popularity profile with parameter $\alpha > 0$ as follows:
   \begin{equation}
   p_{n}=\dfrac{(\frac{1}{n})^\alpha}{\sum_{j=1}^{N}(\frac{1}{j})^\alpha}.
   \label{eq1}
   \end{equation}
   In the following, we first verify the approximation used for finding the optimal configuration of the hybrid scheme. We then study the effect of content popularity, the standard deviation of the number of users in SBSs, and the system scale on the optimum traffic load of the shared link. %To this aim, for Figs. {~\ref{fig2}, ~\ref{vsZStD}}, we suppose that %the number of contents available in the MBS is $N = 1000$, $K = 10$, $M = 100$, and $avg(Z) = 10$. 
%To generate requests in the simulation, each user selects a random number between $0$ and $1$ at each time slot, then the corresponding content to this random number based on the cumulative distribution function of the content popularity is selected, and the user generates a request for this content. In other words, at each time slot, $Z_c$ random numbers are generated for SBS $c$, $\forall c \in \{1,\ldots,K\}$  to determine requests. 

\figurename{~\ref{VsConfigDevZ}} shows the simulation results of MBS traffic load versus $N_1$ for different scenarios of the hybrid, two-partitioning pure coded \cite{hachem2015effect, li2017traffic, ji2017order, zhang2018coded, 8863425} and pure uncoded caching schemes. %In \figurename{~\ref{VsConfig}}, each SBS serves 10 users ($Z_c=10$ $\forall c \in \{1,\ldots,K\}$ ), and in \figurename{~\ref{VsConfigDevZ}}, SBSs serve different numbers of users as $Z_c=1, 3, 5, 7, 9, 11, 13, 15, 17, 19, \forall c \in \{1,\ldots,10\}$ respectively. The content popularity follows the Zipf distribution with $\alpha = 1$. 
As can be understood from \eqref{optimization1} and \figurename{~\ref{VsConfigDevZ}}, the hit ratio of local cache improves as $M_1$ increases. But on the contrary, because of the reduction of memory space for the coded section, the bandwidth load required to satisfy the coded content requests increases. Moreover, by increasing the value of $N_1$, although the contents that are not cached decrease, the required bandwidth load to satisfy requests of coded contents rises.  
%The optimal configuration obtained from the optimization problem for uniform user distribution $(M_1^* \!=\! 37$ and $N_1^* \!=\! 352)$ is indicated in \figurename{~\ref{VsConfig}}. Also, for non-uniform user distribution, the optimal configuration  $(M_1^* \!=\! 40$ and $N_1^* \!=\! 240)$ is indicated in \figurename{~\ref{VsConfigDevZ}}. 
As can be seen in \figurename{~\ref{VsConfigDevZ}}, the optimal configuration found by the optimization problem has the lowest load compared to the other configurations in the simulation. 
%\begin{figure*}[!t]
%\begin{subfigure}{0.49\textwidth}
%	\centering
%	\includegraphics[width=3.2in , height=2 in]{./Figs/bandwidthVsConfig.png}
%	\caption{}
%	\label{VsConfig}
%\end{subfigure}
%\begin{subfigure}{0.49\textwidth}
%	\centering
%	\includegraphics[width=3.2in , height=2 in]{./Figs/bandwidthVsConfigDevZ.png} 
%	\caption{}
%	\label{VsConfigDevZ}
%\end{subfigure}
%\caption{MBS Traffic load as a function of $N_1$ for different $M_1$ and the calculated optimal point for when (a) uniform user distribution in SBSs (each SBS serves $Z_c=10$ users ) and (b) non-uniform user distribution in SBSs.}
%\label{fig2}
%\vspace{-0.8em}
%\end{figure*}
%\begin{figure*}[!t]
%	\centering
%	\begin{subfigure}{0.49\textwidth}
%		\includegraphics[width=3.2in, height=2 in]{./Figs/RvsPopularity.png}
%		\caption{}
%		\label{vsPopularity}
%	\end{subfigure}
%	\begin{subfigure}{0.49\textwidth}
%		\includegraphics[width=3.2in, height=2 in]{./Figs/RvsPopularityDevZ.png}
%		\caption{}
%		\label{vsPopularityDevZ}
%	\end{subfigure}
%	\caption{MBS Traffic load as a function of popularity parameter (a) uniform user distribution in SBSs (each SBS serves $Z_c=10$ users) and (b) non-uniform user distribution in SBSs.}
%	\label{fig3}
%	\vspace{-0.8em}
%\end{figure*}
\begin{figure*}[!t]
	\begin{subfigure}{0.49\textwidth}
		\centering
		\includegraphics[width=3.2in , height=2 in]{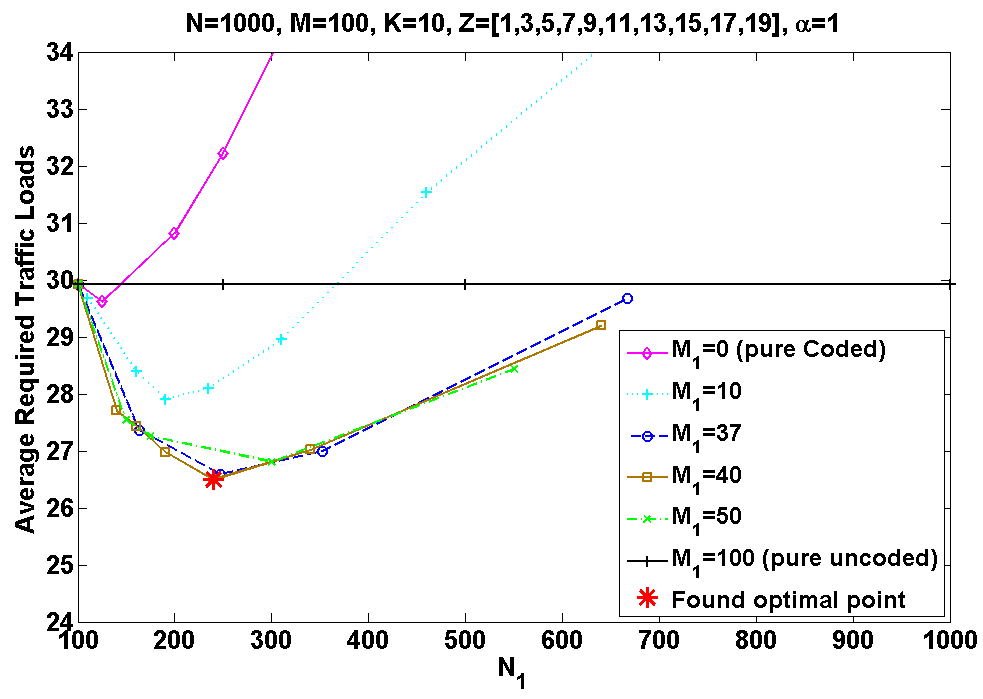}
		\caption{}
		\label{VsConfigDevZ}
	\end{subfigure}
	\begin{subfigure}{0.49\textwidth}
		\centering
		\includegraphics[width=3.2in , height=2 in]{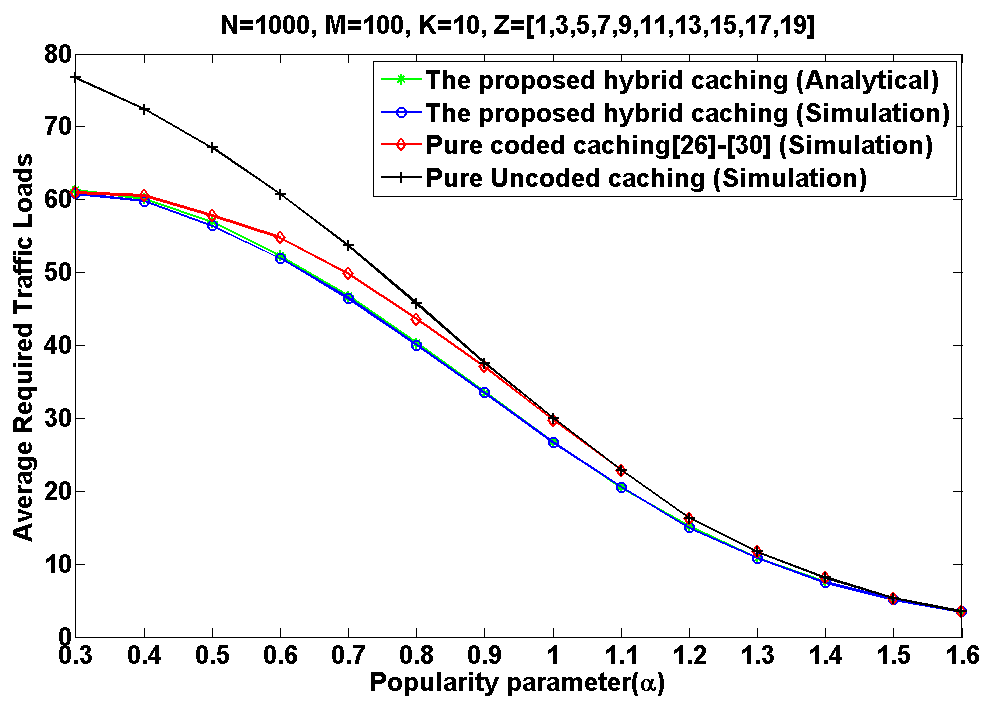} 
		\caption{}
		\label{vsPopularityDevZ}
	\end{subfigure}
	\caption{MBS Traffic load as a function of (a) $N_1$ for different $M_1$ (b) popularity parameter.}
	\label{fig2}
	\vspace{-0.8em}
\end{figure*}
\begin{table}
	\centering
	\small
	\caption{Optimal configuration ($N_1^*$ and $M_1^*$) of the hybrid scheme for different users distribution in SBSs, where $K=10, N=1000, M=100, \alpha=1$ and there are 100 users in the system.}
	\label{table:StD}
	\begin{tabular}{|l|l|l|}%{l l}
		\hline 
		No. of users ($Z_1$, ... ,$Z_{10}$)& $\sigma$(Z)& $N_1^*$, $M_1^*$ \\ 
		\hline 
		10 10 10 10 10 10 10 10 10 10 &$0.0$&$352$, $37$\\ 
		8 \ \ 9 \ \ 9 \  9 \  9 \  10 11 11 12  12 & $1.4142$ &$344$, $39$ \\ 
		6 \ \ 8 \ \ 9 \  9 \  9 \  10 11 12 12  14&$2.3094$&$340$, $40$\\
		5 \ \ 7 \ \ 9 \  9 \  9 \  10 11 12 13  15&	$2.9059$ &$332$, $42$\\ 
		4 \ \  6 \ \  9 \  9 \  9 \  10 11 12 14  16 & $3.5277$&$328$, $43$ \\
		3 \ \  5 \ \  7 \  9 \  9 \  11 11 13 15  17&$4.3461$&$316$, $46$\\ 
		2 \ \ 4 \ \ 6 \  8 \  9 \  11 12 14 16  18&$5.1854$ &$240$, $40$\\ 
		1 \ \ 3 \ \ 5 \ 7 \ 9 \ 11  13 15 17  19 &$6.0553$ &$240$, $40$\\
		0 \ \ 2 \ \ 4 \ 6 \ 9 \ 11 14 16 18  20 &$6.9442$ &$233$, $43$\\ 
		0 \ \ 2 \ \ 2 \ 3 \ 7 \ 11 14 16 20  25 &$8.5894$&$219$, $49$ \\ 
		1 \ \ 1 \ \ 1 \ 1 \ 1 \ 5 \ 15 20 25  30&$11.4504$ &$172$, $52$\\		
		\hline 
	\end{tabular} 
\end{table}	
%\begin{figure}[t]
%	\centering
%	\includegraphics[width=3.2in, height=2 in]{./Figs/RvsZStD.png}
%	\caption{MBS Traffic load as a function of standard deviation of $Z$ when average of $Z$ equals to $10$.}
%	\label{vsZStD}
%	\vspace{-0.8em}
%\end{figure}
\begin{figure*}[!t]
	\centering
%	\begin{subfigure}{0.49\textwidth}
%		\includegraphics[width=\textwidth, height=2 in]{./Figs/new100010010.png}
%		\caption{}
%		\label{fig4a}
%	\end{subfigure}
%	~ 
%	\begin{subfigure}{0.49\textwidth}
%		\includegraphics[width=\textwidth, height=2 in]{./Figs/new10000100010.png}
%		\caption{}
%		\label{fig4b}
%	\end{subfigure}
%	
%	\begin{subfigure}{0.49\textwidth}
%		\includegraphics[width=\textwidth, height=2 in]{./Figs/new100001000100.png}
%		\caption{}
%		\label{fig4c}
%	\end{subfigure}
	\begin{subfigure}{0.49\textwidth}
			\includegraphics[width=\textwidth, height=2 in]{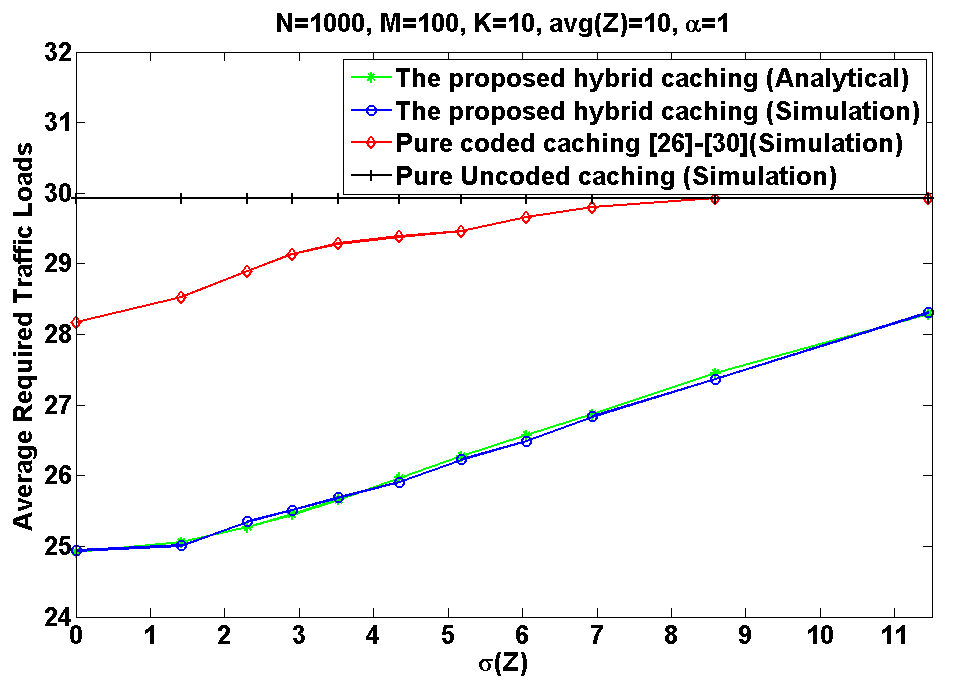}
			\caption{}
			\label{vsZStD}
		\end{subfigure}
	~  
	\begin{subfigure}{0.49\textwidth}
		\includegraphics[width=\textwidth, height=2 in]{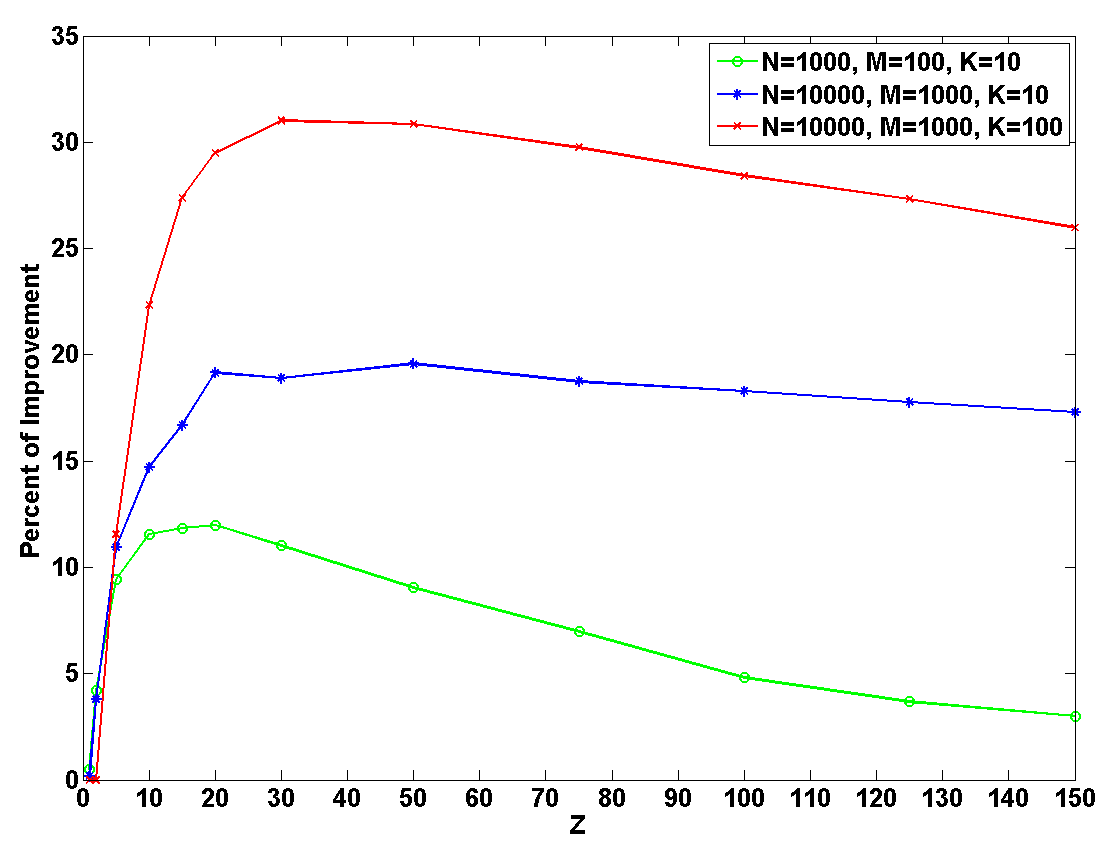}
		\caption{}
		\label{fig4d}
	\end{subfigure}
%	\caption{(a, b, c) MBS Traffic load and (d) percent of improvement of the hybrid caching compared to the two-partitioning pure coded methods \cite{hachem2015effect, li2017traffic, ji2017order, zhang2018coded, 8863425}, as a function of the number of users per SBS $(Z)$ for different system scales.}
	\caption{(a) MBS traffic load as a function of standard deviation of $Z$  (b) percent of improvement (offloading MBS traffic load) of the proposed hybrid caching compared to the two-partitioning pure coded methods \cite{hachem2015effect, li2017traffic, ji2017order, zhang2018coded, 8863425}. }%for different system scales.}
	\label{fig4}
	\vspace{-0.8em}
\end{figure*}
In Figs.~\ref{vsPopularityDevZ} and %{~\ref{fig3}, ~\ref{vsZStD},
~\ref{fig4}, for the hybrid and pure coded schemes \cite{hachem2015effect, li2017traffic, ji2017order, zhang2018coded, 8863425}, we first find the $M_1^*$ and $N_1^*$ values for each parameter settings that minimize the MBS traffic load. The caching schemes are then evaluated for their corresponding optimal configuration via simulation. 
\figurename{~\ref{vsPopularityDevZ}} illustrates the simulation and analytical results for the traffic load as a function of the Zipf parameter in the interval $\alpha \in [0.5, 1.6]$. %for two scenarios of uniform and non-uniform user distribution in the SBSs. 
\tablename{~\ref{table:StD}} shows the optimal configuration of the hybrid caching for some different scenarios of user distribution in the SBSs where there are $100$ users in the system ($K=10$). 
\figurename{~\ref{vsZStD}} also depicts the traffic load of the MBS as a function of the standard deviation of the number of users in the SBSs. This figure shows a comparison of the optimal placement results of different schemes for the configurations of \tablename{~\ref{table:StD}}.
It is evident from these figures that the simulation results are very close to the analytical findings, and the hybrid caching can lead to significant traffic off-loading compared to the two-partitioning pure coded \cite{hachem2015effect, li2017traffic, ji2017order, zhang2018coded, 8863425} and pure uncoded schemes. 
%It can be seen in \figurename{~\ref{fig3}} that the optimal configuration of $M_1 = M$ is achieved when the Zipf parameter ($\alpha$) is larger than $1.5$. Therefore, under such conditions, the optimal traffic load of the two-partitioning pure coded and the hybrid schemes are equal to the traffic load of the uncoded scheme. Also, when $\alpha < 0.5$, the optimal configuration is $M_1^* = 0$ and $N_1^* = N$. In other words, caching with maximum diversity is optimal, and thus, uncoded caching has the worst performance.

As can be seen in Fig. {~\ref{vsZStD}}, %s and{~\ref{fig3}} 
 %when the user distributions in the SBSs are non-uniform, the performance of the coded schemes is reduced. This is because 
 when the standard deviation of the number of users in the SBSs increases, the performance of the coded schemes is reduced. This is because the number of coded requests in the corresponding queues of the MBS becomes very unbalanced. As a result, the expectation of the number of sending steps for a specific configuration increases, and the optimal configuration, as shown in %\figurename{~\ref{fig2}} and 
 \tablename{~\ref{table:StD}}, tends to a smaller quantity of $N_1$ and larger quantities of $M_1$. %However, as can be seen in \figurename{~\ref{VsConfigDevZ}}, the optimal configuration %of uniform user distribution ($N_1=352$, $M_1=37$) 
 %and other configuration near it %$N_1=240$, $M_1=40$ 
 %are also relatively good and perform significantly better than the two-partitioning pure coded and pure uncoded schemes.
%Therefore, even if we do not know the user distributions in the SBSs and only have an estimation of its deviation, by using the hybrid scheme, we can achieve better performance compared to pure methods. 
%\figurename{~\ref{fig4}} shows the traffic load of the MBS versus the number of users within the coverage of each SBS for three different system scales (content library sizes, cache capacities, and number of SBSs). In this figure, the number of users of all SBSs is assumed to be the same and equal to $Z$. The 
%number of users in the system $(K \times Z)$ is taken to be $ 10 \times Z$ for \figurename{~\ref{fig4a}} and \figurename{~\ref{fig4b}}, and $ 100 \times  Z$ for \figurename{~\ref{fig4c}}. Also, \figurename{~\ref{fig4d}} shows the percent of improvement of the hybrid scheme compared to the two-partitioning pure coded methods\cite{hachem2015effect, li2017traffic, ji2017order, zhang2018coded, 8863425} in terms of MBS traffic load for the three system scales.
\figurename{~\ref{fig4d}} shows the percent of improvement of the hybrid scheme compared to the two-partitioning pure coded methods\cite{hachem2015effect, li2017traffic, ji2017order, zhang2018coded, 8863425} in terms of MBS traffic load versus the number of users within the coverage of each SBS for three different system scales (content library sizes, cache capacities, and number of SBSs). In this figure, the number of users of all SBSs is assumed to be the same and equal to $Z$.
As can be seen, the hybrid scheme has made significant improvement in the MBS traffic load, especially for $Z \!>\! 2$. In addition, when $Z$ increases, this percent of improvement increases at first but then decreases. This is because in the hybrid scheme, by increasing $Z$, some requests are hit in the $M_1$ part of SBSs' cache, and therefore the number of steps of sending coded messages becomes significantly less than the pure coded methods. But when $Z$ increases further, the probability of duplicate requests also increases (especially for smaller library sizes). Therefore, the number of steps of sending coded messages in pure coded methods, and therefore the percentage of improvement is reduced.
%
 %In comparison with \figurename{~\ref{fig4a}} and \figurename{~\ref{fig4b}}, we also observe that the traffic load on the MBS of the coded schemes (pure and the hybrid) in \figurename{~\ref{fig4c}} is much less than the pure uncoded scheme. This gap is due to the global caching gain of coded schemes \cite{maddah2014fundamental} and because the number of SBSs is 10 times that of \figurename{~\ref{fig4a}} and \figurename{~\ref{fig4b}}.  
\subsection{ SBS-dependent non-uniform popularity distribution}
In this subsection, we suppose the content popularity distribution is not the same for different SBSs. 
%As mentioned in the previous section, finding the optimal placement in this condition is intractable for all the proposed methods. To this aim, we first evaluate the optimal placement of the hybrid, pure coded, and pure uncoded caching schemes under a small system scale and then evaluate the proposed heuristic methods on larger scales. 
In \figurename{~\ref{optimalheterogenious},  we suppose that the number of available contents in the MBS is $N = 4$, there are a total of $K = 4$ SBSs, each SBS only serves one user ($Z_c=1$ $\forall c \in \{1,\ldots,K\}$), and the content popularity distribution is according to the \tablename{~\ref{table:popularity}}. The results are depicted for various cache capacity of SBSs ($M=1,2$ and $3$). As can be seen in this figure, the hybrid scheme offloads more traffic compared to the two-partitioning pure coded and pure uncoded schemes. 

In particular, when $M=2$, the hybrid scheme is better than both other schemes. In this condition, the best configuration of the hybrid scheme is $N_1^*=3$ and $M_1^*=1$, where ($W_3$) is stored entirely in the caches of $SBS1\&2$, and $W_4$  is stored entirely in the caches of the $SBS3\&4$, also ($W_1, W_2$) are stored with the coded scheme in the caches of all SBSs. Whereas optimal placement of the two-partitioning pure coded method is $N_1^*=4$, or in other words, all $W_1 \dots W_4$ are cached with the coded scheme in the caches of all SBSs. Besides, the optimal placement of pure uncoded scheme is caching $W_1$ (or $W_2$) in all SBSs, $W_3$ in $SBS1\&2$, and $W_4$ in  $SBS3\&4$. In the pure uncoded scheme, one of $W_1$ or $W_2$ is cached in all SBSs, and for the other one, it benefits from multicast opportunities. %In the following, the proposed heuristic methods are evaluated.
\begin{table}
	\centering
	\small
	\caption{Content Popularity distribution for $N=4$ and $K=4$ scenario.}
	\label{table:popularity}
	\begin{tabular}{|l|l|l|l|l|}%{l l}
		\hline 
		&$W_1$& $W_2$& $W_3$,& $W_4$ \\ 
		\hline 
		$SBS1$&0.3 &0.2&0.5& 0.0\\ 
	   $SBS2$&0.2 &0.3&0.5& 0.0\\ 
       $SBS3$&0.3 &0.2&0.0& 0.5\\ 
	  $SBS4$ &0.2 &0.3&0.0& 0.5\\ 
		\hline 
	\end{tabular} 
\end{table}
\begin{figure}[t]
	\centering
	\includegraphics[width=3.2in, height=2 in]{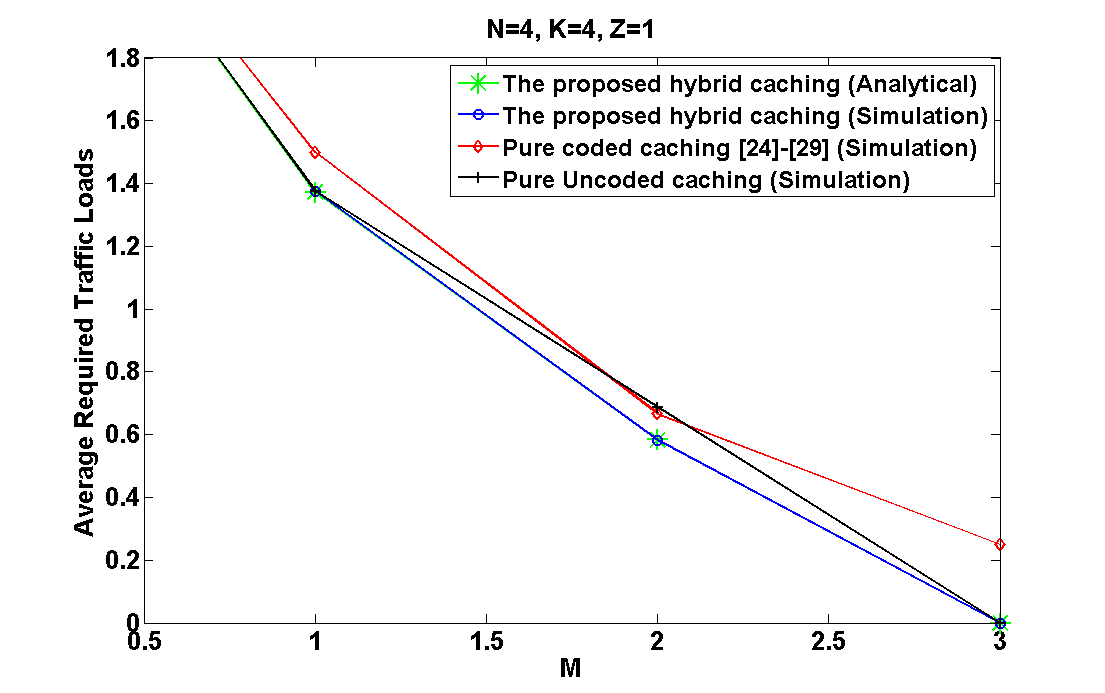}
	\caption{MBS traffic load as a function of M.}
	\label{optimalheterogenious}
	\vspace{-0.8em}
\end{figure}
\section{Conclusion and future works}
\label{sec7}
In this paper, we have studied content caching for the shared medium networks and have proposed a hybrid coded-uncoded caching under heterogeneous users' behaviors. In practice, the proposed scenario corresponds to a cellular network that includes an MBS and multiple SBSs where each SBS is equipped with a limited size cache and serves multiple users. In particular, we assume that each SBS can request a different number of contents than the other ones. Also, we consider non-uniform content popularity distribution, which can be the same for all SBSs (SBS-independent) or different for each one (SBS-dependent). 
We derive explicit closed-form expressions for the server load of the proposed hybrid caching at the delivery phase and formulate the optimum cache partitioning problem. %For the SBS-independent content popularity distributions, the proposed scheme find almost the best cache placement in polynomial complexity. %However, for the heterogeneous non-uniform content popularity distributions, finding the best cache placement of the proposed scheme has exponential complexity. Therefore, in this paper, we proposed several heuristic methods, which among them, $C$-$DH$-$GL$ has the least load on the server. 
Validated by simulation results, our findings showed that the proposed scheme outperforms the baseline schemes of pure uncoded and pure coded caching, as well as the two-partitioning scheme existing in the literature.
For future topics, we first plan to propose heuristic and machine-learning methods in the case of SBS-dependent non-uniform popularity distribution to find appropriate configuration of the proposed hybrid scheme in polynomial complexity. Then, we try to extend this work for online caching problems. 
\appendices
%\begin{appendices}
%\label{Lemma1 Proof}
\section{Lemma1 Proof}
\begin{proof}
Assume that each SBS has one request and all contents are to be cached coded. Then, according to the coded caching scheme in \cite{maddah2014fundamental}, each content is split into ${K} \choose{T}$ non-overlapping fragments, each of size $f/{{K}\choose{T}}$, where $T= \frac{KM}{N}$. Also, each cache selects a $T/K$ fraction of all fragments. 
Then, in each transmission, MBS chooses a new subset of SBSs with size $T+1$, chooses one fragment from the content requested by each SBS, and then, XORs $T+1$ chosen fragments and transmits that. This procedure continues until all possible subsets are chosen. Note that in each transmission, the fragment which is missing at one SBS is available at all other $T$ SBSs ( see the cache placement phase in \cite{maddah2014fundamental}). Consequently, in each transmission, each selected SBS is able to decode one missing fragment using the received signal as well as its cached content. Finally, after ${K} \choose {T+1}$ multicast transmissions, each SBS retrieves its requested content completely. But as mentioned earlier, in step $i$ of coded transmissions in our problem, all SBSs do not have necessarily request, i.e., the coded queues of some SBSs may be empty. In this case, no coded messages are transmitted to those subsets of SBSs that none of their members has a request.     
Hence, if the number of caches that have requests for the coded contents in step $i$ equals $k$, then the number of unnecessary transmissions is ${K-k} \choose {T+1}$. Consequently, the number of multicast transmissions at this step is ${K}\choose {T+1}$ $-$ ${K-k} \choose {T+1}$, where the size of each transmission is equal to $F/{{K} \choose{T}}$ and $T= \frac{K \times (M-M_1)}{(N_1-M_1)}$. The latter is due to the fact that the cache and library sizes are $M-M_1$ and $N_1- M_1$, respectively.

Contrarily, if $(N_1 - M_1) < k$, then the MBS enjoys an improvement of $(N_1 - M_1)/k$ from
broadcast. Thus, from \eqref{eq2}, the traffic load of coded contents is given as:\vspace{-0.4em}
\begin{align}
k \times \Bigg(1 - \frac{M - M_1}{N_1 - M_1}\Bigg) \times \frac{N_1 - M_1}{k} = N_1 - M.
\label{eq4}
\vspace{-0.8em}
\end{align}
This completes the proof.
% that's all folks
\end{proof}
\section{Lemma2Proof}
\begin{proof}
If $l_{c}$ denotes the number of distinct coded requests of SBS $c$, then, $P^{(c)}_i$ equals to $Pr\{l_{c}>=i\}$ and can be calculated according to \eqref{lem2-4}. 
Due to the possibility of duplicate coded requests, the native approach to calculate $Pr\{l_{c}=j\}$ has exponential complexity. However, an approximation of it can be calculated with recursion and dynamic programming in polynomial complexity as fallowing:
if the $Pr\{l^{(z)}_{c}=j\}  $ denotes the probability of having $j$ distinct coded requests in first $z$ requests in SBS $c$, where $z \!=\! 1, 2,\ldots, Z_c$, then $Pr\{l_{c}=j\}$ is equal to  $Pr\{l^{(Z_c)}_{c}=j\}$ and can be calculated with a recursive formula that is given in \eqref{lem2-7}, where the notation $q_j^c$ denotes the probability of requesting the $j$th distinct coded content at the next request in SBS $c$,  where $j-1$ distinct coded contents have been requested until this request. Therefore, $q_j^c$ only depends on the popularity of coded contents and previous $j-1$ distinct requested coded contents. On the other hand, calculating $q_j^c$ is independent of all number of requests. In this section, we assume that the content popularity distribution is SBS-independent in other words, $p_{n,c}=p_n$, $\forall c \in \{1,\ldots,K\}$. Therefore, based on the above definition, $q^c_j$ is the same for all SBSs and hence $q^c_j=q_j$, $\forall c \in \{1,\ldots,K\}$.

Due to the non-uniform popularity distribution of the contents and possibility of duplicate requesting the previous $j-1$ requested coded contents, calculating the exact amount of $q_j$ is highly complicated.
To this aim, we approximate $q_j$ by approximating the probability of requesting $j-1$ previous contents as shown in \eqref{lem2-8}. Apparently, when the popularity distribution is almost uniform, this approximation is more accurate. Contrarily, when the popularity distribution is extremely non-uniform, where few contents are in high demand, this approximation yields lower accuracy and may be calculating a bit larger value for the $r_1$.  However, under extremely non-uniform popularity distributions, the hybrid caching strategy tends to cache these high demand contents entirely (tends to a higher value for $M_1$), and this approximation may reinforce this tendency. By caching the most popular contents entirely, the error of this approximation is greatly reduced. Besides, our goal is to find the best library partitioning policy rather than calculating the exact rate. Therefore, the possible error of such approximation does not affect rate calculation for all contents but only for contents located in different partitions based on different policies.
Hence, as the negligible impact of this approximation on our choice of library partitioning, i.e., $N_1$ and $M_1$, will also be shown numerically and via simulation in the following sections, this approximation is reasonable and has minimal impact on choosing the optimal policy. However, using this approximation is not recommended for selecting the best caching policy for pure coded methods.

This completes the proof.
\end{proof}
	% that's all folks
\section{Lemma3Proof}
\begin{proof}
$\{Q_i=k\}$ occurs when exactly $k$ SBSs receive at least $i$ distinct requests for the coded contents, while the rest of the SBSs receive less than $i$ distinct requests for these contents. Since the number of requests of each SBS may be different from other SBSs, the probability of requesting the coded contents is different for different SBS. Therefore, calculating the $Pr\{Q_i=k\}$ with the native approach is a combinatorial problem with exponential complexity. However, we can  calculate $Pr\{Q_i=k\}$ in polynomial complexity by using recursion and  dynamic programming as follows. If the ${Pr}\{ Q_i^{(c)} = k \}  $ denotes the probability of having $k$ caches with coded requests in first $c$ caches, where $c \!=\! 1, 2,\ldots, K$, then as it is shown in \eqref{lem2-2}, the $Pr\{Q_i=k\}$ is equal to ${Pr}\{ Q_i^{(K)} = k \}$, and it can be calculated with the recursive formula given in \eqref{PrQiK}, where the notation $P^{(c)}_i$ denotes the probability that SBS $c$ (which receives $Z_c$ request at each time slot from $Z_c$ users) has at least $i$ distinct requests for coded contents. 

This completes the proof.
\end{proof}
\section{Lemma4Proof}
\begin{proof}
The proof of lemma\ref{PIgC} is similar to lemma\ref{PIC}, and the proof of equations \eqref{lem5-4}-\eqref{lem5-7} of this lemma is similar to lemma\ref{PIC}, except that in lemma\ref{PIgC}, there are $G$ coded delivery clusters, and therefore these equations are calculated for each one separately. However, for each cluster, only the SBSs that participate in it are considered. 
In addition, in lemma\ref{PIgC}, the $p_{n,c}$ is not the same for all SBSs. Therefore equation \eqref{lem5-8} has some differences from \eqref{lem2-8}. However, likewise lemma\ref{PIC}, $q_{g,j}^c$ for each cluster $g$ only depends on the popularity of $N_g$ contents and previous $j-1$ requests for these contents from corresponding SBS $c$. Also, $q_{g,j}^c$ is calculated by approximating the probability of requesting the previous $j-1$ of $N_g$ contents in SBS $c$. Except that in this lemma, $q_{g,j}^c$ is not the same for different SBSs, and also for each SBS, it is not the same for different clusters. Note that for each cluster $g$ only $K_g$ SBSs that participate in it are considered. Therefore,  $q_{g,j}^c$ is calculated only for SBSs, which are participated in cluster $g$. As shown in \eqref{lem5-8}, The array $X$ is used to determine coded contents corresponding to cluster $g$.

This completes the proof.
\end{proof}
%\end{appendices}
\bibliographystyle{IEEEtran}
\bibliography{journal} 
\end{document}